\documentclass[runningheads, envcountsect, 10pt]{llncs}

\usepackage{amssymb}
\usepackage{url}

\begin{document}
	
	\title{Recognizing the Tractability in Big Data Computing}
	
	\titlerunning{Recognizing the Tractability in Big Data Computing}
	
	\author{Xiangyu Gao\inst{1} \and
		Jianzhong Li\inst{2} \and
		Dongjing Miao\inst{3} \and
		Xianmin Liu\inst{4}
	}
	
	\authorrunning{X. Gao, J. Li, et al.}
	
	\institute{
		\email{gaoxy@hit.edu.cn}\and
		\email{lijzh@hit.edu.cn}\and
		\email{miaodongjing@hit.edu.cn}\and
		\email{liuxianmin@hit.edu.cn}\\
		Harbin Institute of Technology, Harbin, Heilongjiang 150001, China
	}
	
	\maketitle
	
	\begin{abstract}
		Due to the limitation on computational power of existing computers, the polynomial time does not works for identifying the tracta-ble problems in big data computing. 
		This paper adopts the sublinear time as the new tractable standard to recognize the tractability in big data computing, and the random-access Turing machine is used as the computational model to characterize the problems that are tractable on big data. 
		First, two pure-tractable classes are first proposed.
		One is the class $\mathrm{PL}$ consisting of the problems that can be solved in polylogarithmic time by a RATM.
		The another one is the class $\mathrm{ST}$ including all the problems that can be solved in sublinear time by a RATM.
		The structure of the two pure-tractable classes is deeply investigated and they are proved $\mathrm{PL^i} \subsetneq \mathrm{PL^{i+1}}$ and $\mathrm{PL} \subsetneq \mathrm{ST}$.		
		Then, two pseudo-tractable classes, $\mathrm{PTR}$ and $\mathrm{PTE}$, are proposed.
		$\mathrm{PTR}$ consists of all the problems that can solved by a RATM in sublinear time after a PTIME preprocessing by reducing the size of input dataset.
		$\mathrm{PTE}$ includes all the problems that can solved by a RATM in sublinear time after a PTIME preprocessing by extending the size of input dataset.
		The relations among the two pseudo-tractable classes and other complexity classes are investigated and they are proved that $\mathrm{PT} \subseteq \mathrm{P}$, $\sqcap'\mathrm{T^0_Q} \subsetneq \mathrm{PTR^0_Q}$ and $\mathrm{PT_P} = \mathrm{P}$.	
	\end{abstract}

	\section{Introduction}
		
		Due to the limitation on computational power of existing computers, the challenges brought by big data suggest that the tractability should be re-considered for big data computing.
		Traditionally, a problem is tractable if there is an algorithm for solving the problem in time bounded by a polynomial (PTIME) in size of the input.
		In practice, PTIME no longer works for identifying the tractable problems in big data computing. 
		
		\begin{example}
			Sorting is a fundamental operation in computer science, and many efficient algorithms have been developed.
			Recently, some algorithms are proposed for sorting big data, such as Samplesort and Terasort.
			However, these algorithms are not powerful enough for big data since their time complexity is still $O(n\log{n})$ in nature. 
			We performed Samplesort and Terasort algorithms on a dataset with size 1 peta bytes, 1PB for short.
			The computing platform is a cluster of 33 computation nodes, each of which has 2 Intel Xeon CPUs, interconnected by a 1000 Mbps ethernet.
			The Samplesort algorithm took more than 35 days, and the Terasort algorithm took more than 90 days.
			
		\end{example}

		\begin{example}
			Even using the fastest solid state drives in current market, whose I/O bandwidth is smaller than 8GB per second \cite{PCIewiki}, a linear scan of a dataset with size 1 PB, will take 34.7 hours.
			The time of the linear scan is the lower bound of many data processing problems. 
		\end{example}
		
        Example 1 shows that PTIME is no more the good yardstick for tractability in big data computing.  
        Example 2 indicates that the linear time, is still unacceptable in big data computing.

		This paper suggests that the sublinear time should be the new tractable standard in big data computing.
		Besides the problems that can be solved in sublinear time directly, many problems can also be solved in sublinear time by adding a one-time preprocessing.
		For example, searching a special data in a dataset can be solved in $O(\log{n})$ time by sorting the dataset first, a $O(n\log{n})$ time preprocessing, where $n$ is the size of the dataset.
		
		To date, some effort has been devoted to providing the new tractability standard on big data.
		In 2013, Fan et al. made the first attempt to formally characterize query classes that are feasible on big data.
		They defined a concept of $\sqcap$-tractability, $i.e.$ a query is $\sqcap$-tractable if it can be processed in parallel polylogarithmic time after a PTIME preprocessing \cite{Fan2013Making}.
		Actually, they gave a new standard of tractability in big data computing, that is, a problem is tractable if it can be solved in parallel polylogarithmic time after a PTIME preprocessing.
		They showed that several feasible problems on big data conform their definition.
		They also gave some interesting results.
		This work is impressive but still need to be improved for the following reasons.		
		(1) Different from the traditional complexity theory, this work only focuses on the problem of boolean query processing.
		(2) The work only concerns with the problems that can be solved in parallel polylogarithmic time after a PTIME preprocessing.
		Actually, many problems can be solved in parallel polylogarithmic time without PTIME preprocessing.
		(3) The work is based on parallel computational models or parallel computing platforms without considering the general computational models.		
		(4) The work takes polylogarithmic time as the only standard for tractability. Polylogarithmic time is a special case of the sublinear time, and it is not sufficient for characterizing the tractbility in big data computing. 
			
		Similar to the $\sqcap$-tractability theory \cite{Fan2013Making}, Yang et al. placed a logarithmic-size restriction on the preprocessing result and relaxed the query execution time to PTIME and introduced the corresponding  $\sqcap'$-tractability \cite{Yang2017Tractable}.
		They clarified that a short query is tractable if it can be evaluated in PTIME after a one-time preprocessing with logarithmic-size output. 
		This work just pursued Fan et al.'s methodology, and there is no improvement on Fan et al.'s work \cite{Fan2013Making}. 
        Besides, the logarithmic restriction on the output size of preprocessing is too strict to cover all query classes that are tractable on big data.
        
        In addition, computation model is the fundamental in the theory of computational complexity.
        Deterministic Turing machine (DTM) is not suitable to characterize sublinear time algorithms since the sequential operating mode makes only the front part of input can be read in sublinear time.
        For instance, searching in an ordered list is a classical problem that can be solved in logarithmic time.
        However, if DTM is used to describe the computation procedure of it, the running time comes to polynomial.
        To describe sublinear time computation accurately and make all of the input can be accessed in sublinear time, random access is very significant.
        
		This paper is to further recognize the tractability in big data computing.
		A general computational model, random-access Turing machine, is used to characterize the problems that are tractable on big data, not only query problems.
		The sublinear time, rather than polylogarithmic time, is adopted as the tractability standard for big data computing.
		Two classes of tractable problems on big data are identified.
		The first class is called as pure-tractable class, including the problems that can be solved in sublinear time without preprocessing.
		The second class is called as pseudo-tractable class, consisting of the problems that can be solved in sublinear time after a PTIME preprocessing.
		The structure of the two classes is investigated, and the relations among the two classes and other existing complexity classes are also studied.
		The main contributions of this paper are as follows.
		
		(1) To describe sublinear time computation more accurately, the random-access Turing machine (RATM) is formally defined.
		RATM is used in the whole work of this paper.
		It is proved that the RATM is equivalent to the deterministic Turing machine in polynomial time.
		Based on RATM, an efficient universal random-access Turing machine $\mathcal{U}$ is devised.
		The input and output of $\mathcal{U}$ are $(x, c(M))$ and $M(x)$ respectively, where $c(M)$ is the encoding of a RATM $M$, $x$ is an input of $M$, and $M(x)$ is the output of $M$ on input $x$.
		Moreover, if $M$ halts on input $x$ within $T$ steps, then $\mathcal{U}$ halts within $cT\log{T}$ steps, where $c$ is a constant.
		
		(2) Using RATM and taking sublinear time as the tractability standard, the classes of tractable problems in big data computing are defined. 
		First, two pure-tractable complexity classes are defined. One is a polylogarithmic time class $\mathrm{PL}$, which is the set of problems that can be solved by a RATM in polylogarithmic time.
		The another one is a sublinear time class $\mathrm{ST}$, which is the set of all decision problems that can be solved by a RATM in sublinear time.
		Then, two pseudo-tractable classes, $\mathrm{PTR}$ and $\mathrm{PTE}$, 
		are first defined.
		$\mathrm{PTR}$ is the set of all problems that can be solved in sublinear time after a PTIME preprocessing by reducing the size of input dataset.
		$\mathrm{PTE}$ is the set of all problems that can be solved in sublinear time after a PTIME preprocessing by extending the size of input dataset.
		 
		(3) The structure of the pure-tractable classes is investigated deeply. 
		It is first proved that $\mathrm{PL}^i \subsetneq \mathrm{PL}^{i+1}$, where $\mathrm{PL}^i$ is the class of the problems that can be solved by a RATM in $O(\log^i{n})$ time and $n$ is the size of the input.
		Thus, a polylogarithmic time hierarchy is obtained. 
		It is proved that $\mathrm{PL}=\bigcup_i\mathrm{PL}^i \subsetneq \mathrm{ST}$.
		This result shows that there is a gap between polylogarithmic time class and linear time class. 
		It is also proved that $\mathrm{DLOGTIME}$ reduction \cite{Buss1987The} is closed for $\mathrm{PL}$ and $\mathrm{ST}$.
		The first $\mathrm{PL}$-complete problem and the first $\mathrm{ST}$-complete problem are given also.
		
		(4) The relations among 
		the complexity classes 
		$\mathrm{PTR}$, 
		$\mathrm{PTE}$, ,  
		$\sqcap'\mathrm{T^0_Q}$ \cite{Yang2017Tractable} and $\mathrm{P}$ is studied.
		They are proved that $\mathrm{PT} \subseteq \mathrm{P}$ and $\sqcap'\mathrm{T^0_Q} \subsetneq \mathrm{PTR^0_Q}$.
		Finally, it is concluded that all problems in $\mathrm{P}$ can be made pseudo-tractable.
		
		The remainder of this paper is organized as follows.
		Section 2 formally defines the complexity model RATM, proves that RATM is equivalent to DTM and there is an efficient URATM, and defines the problem in big data computing. 
		Section 3 defines the pure-tractable classes, and investigates the structure of the pure-tractable classes.
		Section 4 defines the pseudo-tractable classes, and studies the relations among the complexity classes 
		$\mathrm{PTR}$, 
		$\mathrm{PTE}$,  
		$\sqcap'\mathrm{T^0_Q}$ \cite{Yang2017Tractable} and $\mathrm{P}$.
		Finally, Section 5 concludes the paper.

	\section{Preliminaries}
		
		To define sublinear time complexity classes precisely, a suitable computation model should be chosen since sublinear time algorithms may read only a miniscule fraction of its input and thus random access is very important. 
		The random-access Turing machine is chosen as the foundation of the work in this paper.
		This section gives the formal definition of the random-access Turing machine.
		They are proved that the RATM is equivalent to the determinate Turing machine (DTM) in polynomial time and there is a universal random-access Turing machine.
		Finally, a problem in big data computing is defined.
	
		\subsection{Random-access Turing Machine}
		
			\noindent\textbf{random-access Turing machine} 
			A random-access Turing machine (RATM) $M$ is a $k$-tape Turing machine with an input tape and an output tape and is additionally equipped with $k$ binary index tapes that are write-only.
			One of the binary index tapes is for $M$'s read-only input tape and the others for the $M$'s $k-1$ work tapes.
			Note that $k \ge 2$.
			The formally definition of RATM is as follows. 
			
			\begin{definition}
				A \emph{RATM} $M$ is a 8-tuple $M = (Q, \Sigma, \Gamma, \delta, q_0, B, q_f, q_a)$, where
				
				$Q$: The finite set of states.
				
				$\Sigma$: The finite set of input symbols.
				
				$\Gamma$: The finite set of tape symbols, and $\Sigma \subseteq \Gamma$. 
				
				$\delta$: $Q \times \Gamma^k \to Q \times \Gamma^{k-1} \times \{0, 1, B\}^k \times \{L, S, R\}^{2k}$, where $k \ge 2$.
				
				$q_0 \in Q$: The start state of $M$.
				
				$B \in \Gamma \setminus \Sigma$: The blank symbol.
				
				$q_f \in Q$: The accepting state.
				
				$q_a \in Q$: The random access state.
				If $M$ enters state $q_a$, $M$ will move the heads of all non-index tapes to the cells described by the respective index tapes automatically.
			\end{definition}
			
			Assuming the first tape of a RATM $M$ is the input tape, if $M$ is in state $q \in Q$, $(\sigma_1, \cdots, \sigma_k)$ are the symbols currently being read in the $k$ non-index tapes of $M$, and the related transition function is $\delta(q, (\sigma_1, \cdots, \sigma_k)) = (q', (\sigma'_2, \cdots, \sigma'_k), (a_1, \cdots\\, a_k), (z_1, \cdots, z_{2k}))$,
			$M$ will replace $\sigma_i$ with $\sigma'_i$, where $2 \le i \le k$, write $a_j$ $(1 \le j \le k)$ on the corresponding index tape, move heads Left, Right, or Stay in place as given by $(z_1, \cdots, z_{2k})$, and enter the new state $q'$.
			
			The following lemmas state that RATM is equivalent to the deterministic Turing machine (DTM).
			
			\begin{lemma}\label{lema:ratm-tm}
				For a Boolean function $f$ and a time-constructible function \emph{\cite{arora2009computational}} $T$,   
				
				\emph{(1)} if $f$ is computable by a \emph{DTM} within time $T(n)$, then it is computable by a \emph{RATM} within time $T(n)$, and
				
				\emph{(2)} if $f$ is computable by a \emph{RATM} within time $T(n)$, then it is computable by a \emph{DTM} within time $T(n)^2\log{T(n)}$.
			\end{lemma}
		
			\begin{proof}
				(1) is easy to be proved since RATM can simulate DTM step by step through omitting the random access ability of RATM.
				To prove (2), we can construct a $2k$-tape DTM $M$ to simulate a $k$-tape RATM $N$.
				$M$ uses $k$ tapes to simulate the $k$ index tapes of $N$, and the other $k$ tapes to simulate the $k$ non-index tapes of $N$.
				If the contents on a non-index tape $T$ of $N$ is $c_1, \cdots, c_j$, 
				the corresponding tape of $M$ will contain $* a_1 \# c_1 * a_2 \# c_2 \cdots * a_j \# c_j$, where $a_i$ is the address of $c_i$ on tape $T$ of $N$.
				Since $N$ stops in $T(n)$ steps, there are at most $T(n)$ non-blank symbols on each tape of $N$.
				Therefore, the length of the corresponding tape of $M$ is $\Sigma_{i=1}^{T(n)} (\log{i} + 3) = \Theta(T(n)\log{T(n)})$.
				$M$ simulates $N$ as follows.
				
				(1) If $N$ does not enter the random access state $q_a$, $M$ just acts like $N$;
				
				(2) If $N$ enters $q_a$, then $M$ first moves the heads of its $k$ non-index tapes to the leftmost, then moves from left to right to find the symbol $c_i$, where $a_i$ is equal to the address on the corresponding index tape of $N$. 
				
				Since the maximum length of $M$'s tapes is $T(n)\log{T(n)}$ and the running time of $N$ is $T(n)$, the running time of $M$ is at most $T(n)^2\log{T(n)}$.
				\qed
			\end{proof}
			
			\begin{corollary}\label{coro:ratm-tm}
				If a Boolean function $f$ is computable by a \emph{RATM} within time $o(n)$, then it is computable by a \emph{DTM} within time $n\log{n}o(n)$.
			\end{corollary}
		
			\begin{proof}
				$M$ simulates $N$ in the same way as in the proof of lemma \ref{lema:ratm-tm}. 
				Since the maximum length of $M$'s tapes is $n\log{n}$ when the runtime of $N$ is $o(n)$ and the running time of $N$ is $o(n)$, the running time of $M$ is at most $o(n)n\log{n}$.
				\qed
			\end{proof}
	
		\subsection{The Universal Random-access Turing Machine}
		
			Just like DTM, RATM can be encoded by a binary string.
			The code of a RATM $M$ is denoted by $c(M)$.
			The encoding method of RATM is the same as that of DTM \cite{du2011theory}.	
			The encoding of RATM makes it possible to devise a universal random-access Turing machine (URATM) with input $(x, c(M))$ and outputs $M(x)$, where $x$ is the input of a RATM $M$, $c(M)$ is the code of $M$, and $M(x)$ is the output of $M$ on $x$.
			Before the formal introduction of URATM, we first present two lemmas.
			
			\begin{lemma}\label{lema:ratm-alph}
				For every function $f$, if $f$ is computable in time $T(n)$ by a \emph{RATM} $M$ using alphabet $\Gamma$, then it is computable in time $c_1T(n)$ by a \emph{RATM} $\tilde{M}$ using alphabet $\{0, 1, B\}$.
			\end{lemma}
		
			\begin{proof}
				Let $M = (Q, \Sigma, \Gamma, \delta, q_0, B, q_f, q_a)$ be a $k$-tape RATM which computes $f$ in time $T(n)$.
				We define a $2k$-tape RATM $\tilde{M} = (Q', \{0, 1\}, \{0, 1, B\}, \delta', q'_0, B, q'_f, q'_a)$ in the follows to compute $f$ in time $cT(n)$.
				
				Let $M$'s non-index tapes and index tapes be numbered from $1$ to $k$, $\tilde{M}$'s non-index tapes, and index tapes be numbered from $1$ to $2k$.
				Let $b$ is the least number satisfying $2^b \ge \log{|\Gamma|} \ge 2^{b-1}$ bits.
				Every symbol of $\Gamma$ be encoded using binary code with length $2^b$. 
				
				Thus, the $j$th non-index tape of $\tilde{M}$ simulates the $j$th non-index tape of $M$ using the binary codes above for $1 \le j \le k$, 
				that is, there are $2^b$ cells in $\tilde{M}$'s $j$th non-index tape for every cell in $M$'s $j$th non-index tape.
				
				$\tilde{M}$ utilizes a non-index tape and an index tape to simulate an index tape of $M$.
				More precisely, $\tilde{M}$'s $(j+k)$th non-index tape is used to make a calculation on the contents of $M$'s $j$th index tape.
				And $\tilde{M}$'s $j$th index tape works like the $j$th index tape of $M$.
				The concrete operations will be introduced later.
				If the contents of $M$'s $j$th index tape of $M$ is $c_j$, the contents of $\tilde{M}$'s $(j+k)$th non-index tape and $j$th index tape are $c_j 0^b$.
				And let $Q' = Q \times \{0, 1, B\}^{k \times 2^b}$, $q'_0 = (q_0, B^{k \times 2^b})$, $q'_f = (q_f, B^{k \times 2^b})$ and $q'_a = (q_a, B^{k \times 2^b})$.
				
				The initial configuration of $\tilde{M}$ is as follows.
				
				(1) Input $x$, which is encoded, is stored in input tape, and the state of $\tilde{M}$ is $q'_0$.
				
				(2) $0^b$ is stored in $\tilde{M}$'s last $k$ non-index tapes and first $k$ index tapes.
				
				To simulate one step of $M$ from state $q$, $\tilde{M}$ starts from state $(q, B^{k \times 2^b})$ and acts as follows.
				
				(1) $\tilde{M}$ uses $2^b$ steps to read from its first $k$ non-index tapes.
				After that, the state of $\tilde{M}$ becomes $(q, c(\sigma_1), \cdots, c(\sigma_k))$, where $c(\sigma_i)$ denotes the code of $\sigma_i$.
				
				Assume that there is a transition function of $M$ is $\delta(q, (\sigma_1, \cdots, \sigma_k)) = (q', (\sigma'_2,\\ \cdots, \sigma'_k), (a_1, \cdots, a_k), (z_1, \cdots, z_{2k}))$.
				
				(2) $\tilde{M}$ uses $2^b$ steps to write $c(\sigma_2), \cdots, c(\sigma_k)$ on its first $k$ non-index tapes except the input tape since the input tape is read-only.
				
				(3) $\tilde{M}$ uses $2b$ steps to write $a_1, \cdots, a_k$ on its index tapes.
				More precisely, 
				if $a_j = 0$ or $1$, $\tilde{M}$  moves the head of the $(j+k)$th non-index tape $a$ locations to the right and checks whether there is $0^{b-1}B$.
				If so, it writes $0$, and moves $b$ locations to the left and writes $a_j$.
				If not, it moves $b$ locations to the left and writes $a_j$. 
				If $a_j = B$, $\tilde{M}$ moves the head of the $(j+k)$th non-index tape $b+1$ locations to the right and checks whether there is $0^bB$.
				If so, it moves one location to the left and writes $B$, then it moves $b$ locations to the left and writes $0$.
				
				(4) $\tilde{M}$ uses $2^b$ steps to move the heads of non-index tapes according to $z_1, \cdots, z_k$ and uses one step to move the heads of index tapes according to $z_{k+1}, \cdots, z_{2k}$.
				
				(5) Then $\tilde{M}$ enters $(q', B^{k \times 2^b})$.
				
				Note that $\tilde{M}$ uses $3 \times 2^b + 2b$ steps to simulate one step of $M$.
				Thus, the total number of steps of $\tilde{M}$ is at most $c_1T(n)$ where $c_1$ is a constant depending on the size of the alphabet.
				\qed
			\end{proof}
		
			\begin{lemma}\label{lema:ratm-tape}
				For every function $f$, if $f$ is computable in time $T(n)$ by a \emph{RATM} $M$ using $k$ tapes, then it is computable in time $c_2T(n)\log{T(n)}$ by a 5-tape \emph{RATM} $\tilde{M}$. 
			\end{lemma}
		
			\begin{proof}
				If $M = (Q, \Sigma, \Gamma, \delta, q_0, B, q_f, q_a)$ is a $k$-tape RATM that computes $f$ in $T(n)$ time, we describe a $5$-tape RATM $\tilde{M} = (Q', \Sigma, \Gamma', \delta', q'_0, B', q'_f, q'_a)$ computing $f$ in time $c_2T(n)\log{T(n)}$.
				
				$\tilde{M}$ uses its input tape and output tape in the same way as $M$ does.
				The first work tape of $\tilde{M}$ named main work tape is used to simulate the contents of $M$'s $k-2$ work tape.
				The second work tape of $\tilde{M}$ named main index tape is used to simulate the contents of $M$'s $k-2$ index tape of $k-2$ work tapes.
				The last work tape of $\tilde{M}$ name usual movement tape is used to store the positions of heads of $M$'s $k-2$ work tapes.
				Note that each of the three work tapes has $k-2$ tracks, each which simulates one work tape of $M$.
				
				The symbol in one cell of $\tilde{M}$'s main work tape is in $\Gamma^{k-2}$, each corresponding to a symbol on a work tape of $M$.
				The symbol in one cell of $\tilde{M}$'s main index tape and usual movement tape is in $\{0, 1, B\}^{k-2}$, each corresponding to a symbol on an index tape of $M$.
				Each track of the main index tape has a symbol in $\{\hat{0}, \hat{1}, \hat{B}\}$ to indicate the head position of corresponding index tape.
				Each track of the usual movement tape records the head's location of the corresponding work tape.
				Hence, we have $Q' = Q \times \Gamma^k$, $\Gamma' = \Gamma \cup \{\Gamma \cup \hat{\Gamma}\}^{k-2}$, $q'_0 = (q_0, B^k)$, $B' = \{B, B^{k-2}\}$, $q'_f = (q_f, B^k)$ and $q'_a = (q_a, B^k)$.
				
				The initial configuration of $\tilde{M}$ is as follows.
				
				(1) Input $x$ is stored in the input tape, and the state of $\tilde{M}$ is $q'_0$
				
				(2) Main work tape, main index tape and output tape are empty.
				
				(3) $0^{k-2}$ is stored in the usual movement tape since the $k-2$ heads of $M$'s work tapes are at the first blank symbol whose index is $0$.
				
				For a computational step of $M$ in state $q$, $\tilde{M}$ starts from the state $(q, B^k)$ to simulate $M$ as follows.
				
				(1) Read $\sigma_1$ and $\sigma_k$ from input and output tapes respectively.
				
				(2) $\tilde{M}$ uses $k-2$ times reads to simulate the parallel reads  of $M$ to gather the symbols on $k-2$ work tapes of $M$:
				
				(2.1) Read the head positions on the $i$th track from the usual movement tape and stores them into the index tape of the main work tape;
				
				(2.2) Enter the random access state $q'_a$ to read the corresponding symbol $\sigma_i$ of the $i$th track from the main work tape for $1 \leq i \leq k-2$;
				
				After that, $\tilde{M}$ is in state $(q, \sigma_1, \cdots, \sigma_k)$.
				
				If there is a transition function of $M$, $\delta(q, (\sigma_1, \cdots, \sigma_k)) = (q', (\sigma'_2, \cdots, \sigma'_{k}),\\ (a_1, \cdots, a_k), (z_1, z_2, \cdots, z_{2k}))$, then  $\tilde{M}$ continues its work as follows.
				
				(3) For the input tape, $\tilde{M}$ moves the head according to $z_1$, writes $a_1$ on the index tape and moves the heads of index tape according to $z_{k+1}$.
				For the output tape, $\tilde{M}$ writes $\sigma'_k$ and moves the head according to $z_k$, writes $a_k$ on the index tape and moves the heads of index tape according to $z_{2k}$.
				
				(4) $\tilde{M}$ uses $k-2$ times writes to write $(\sigma'_2, \cdots, \sigma'_{k-1})$ to the main work tape and modify usual movement tape:
				
				(4.1) Read $i$th track from the usual movement tape into the index tape of the main work tape,
				
				(4.2) Enter the random access state $q'_a$ and then write the symbol $\sigma'_{i+1}$ to the cell, and
				
				(4.3) If $z_{i+1}$ is $L$, reduce $1$ from the number on the $i$th track of the usual movement tape.
				If $z_{i+1}$ is $R$, add $1$ to the number on the $i$th track of the usual movement tape.
				If $z_{i+1}$ is $S$, do nothing.
				
				(5) $\tilde{M}$ uses $k-2$ times writes to write $(a_2, \cdots, a_{k-1})$ to the main index tape:
				
				(5.1) Scan main index tape from left to right to find the head position $p$ on the $i$th track, then writes the $a_{i+1}$ on it.
				
				(5.2) If $z_{k+i+1}$ is $L$, replace the symbol $a$ before $p$ with $\hat{a}$.
				If $z_{k+i+1}$ is $R$, replace the symbol behind $p$ with $\hat{a}$.
				If $z_{k+i+1}$ is $S$, replace the symbol $a$ at $p$ with $\hat{a}$.
				
				(6) Change state to $(q', B^k)$. Then, if $q'$ is $q_a$, copy the main index tape to the usual movement tape.
				
				Now, we analyze the number of steps that $\tilde{M}$ uses to simulate one step of $M$.
				Since the running time of $M$ is $T$, we can assume that the max length of work tapes $M$ used in computation is $T$ without loss of generality.
				Thus, the max length of indexes written on the index tape by $M$ is $\log{T}$.
				$\tilde M$ uses the main work tape to simulate $k-2$ work tapes of $M$,  and uses the main index tape and the usual movement tape to simulate $k-2$ index tapes of $M$ and $k-2$ head positions of $M$'s work tapes.
				Thus, the length of main work tape is at most $T$ and the length of main index tape and usual movement tape are at most $\log{T}$.
				Recall the operations mentioned above, read a track on the usual movement tape and write to the index tape take at most $\log{T}$ steps.
				The random access takes only one step.
				And the modification of the usual movement tape takes at most $\log{T}$ steps.
				And the copy from the main index tape to the usual movement tape takes at most $\log{T}$ steps.
				It can be seen that $\tilde{M}$ uses $c_2\log{T}$ steps to simulate one steps of $M$. Since $M$ takes $T$ steps in total, and thus the total number of steps taken by $\tilde{M}$ is $c_2T\log{T}$, where $c_2$ is a constant depending on the number of tapes.
				\qed
			\end{proof}
		
			\begin{theorem}\label{thrm:uratm}
				There exists a universal random-access Turing machine $\mathcal{U}$, whose input is $(x, c(M))$ and outputs is $M(x)$, where $x$ is an input of $M$, $c(M)$ is the code of $M$, and $M(x)$ is the output of $M$ on $x$. 
				Moreover, if $M$ halts on input $x$ in $T$ steps then $\mathcal{U}$ halts on input $(x, c(M))$ in $O(cT\log{T})$ steps, where $c$ is a constant depending on $M$.
			\end{theorem}
			
			\begin{proof}
				From Lemma \ref{lema:ratm-alph} and Lemma \ref{lema:ratm-tape}, we only need to design a URATM $\mathcal{U}$ to simulate any 10-tape RATM. 
				Let $\mathcal{U}$ be a $12$-tape RATM defined as
				\begin{eqnarray*}
					\mathcal{U}=(Q, \{0, 1\}, \{0, 1, B\}, \delta, q_0, B, q_f, q_a),
				\end{eqnarray*}
				where the input of $\mathcal{U}$ is $(x, c(M))$, $x$ is an input of a 10-tape RATM $M$ with alphabet $\{0, 1, B\}$, and $c(M)$ is the code of $M$.	
				
				$\mathcal{U}$ will use its input tape, output tape and the first eight work tapes in the same way $M$ does.
				In addition, the transition functions of $M$ are stored in the first extra work tape of $\mathcal{U}$.
				The current state of $M$ and symbols read by $M$ are stored in the second extra work tape of $\mathcal{U}$.
				
				$\mathcal{U}$ simulates one computational step of $M$ on input $x$ as follows.
				
				(1) $\mathcal{U}$ stores the state of $M$, and the symbols read from input tape, output tape and the first eight work tape to the second extra work tape.
				
				(2) $\mathcal{U}$ scans the table of $M$'s transition function stored in the first extra work tape to find out the related transition function.
				
				(3) $\mathcal{U}$ replaces the state stored in the second work tape to the new state, writes symbols and moves heads.
				If the new state is random access state, then $\mathcal{U}$ enters $q_a$.
				
				Now, we analyze the number of steps that $\mathcal{U}$ uses to simulate one step of a 10-tape RATM $M$.
				In step (1), $\mathcal{U}$ takes one step to read symbols from its input, output and the first eight work tapes.
				In step (2), $\mathcal{U}$ makes a linear scan to find the related transition function, and it takes $c_3$ steps, where $c_3$ is the size of the transition function of $M$.
				In step (3), $\mathcal{U}$ uses two steps to write symbols to its input, output and the first eight work tapes and moves heads of them.
				Since $M$ halts on input $x$ within $T$ steps, then $\mathcal{U}$ will halts in $(c_3+3)T$ steps.
				
				Since any $k$-tape RATM that stops in $T$ steps can be simulated by a 10-tape RATM with alphabet $\{0, 1, B\}$ in $c'T\log{T}$ steps, any $k$-tape RATM can be simulated by $U$ in $(c_3+3)(c'T\log{T})$ steps, i.e. $O(cT\log{T})$, where $c=c_3c'$ depending on $M$. 
				\qed
			\end{proof}
						
			Theorem \ref{thrm:uratm} is an encouraging result which can help us to investigate the structure of sublinear time complexity classes.
			
		\subsection{Problems in Big Data Computing}
		
			To reflect the characteristics in big data computing, a problem in big data computing is defined as follows.
			
			INPUT: big data $D$, and a function $\mathcal{F}$ on $D$.
			
			OUTPUT: $\mathcal{F}(D)$.
			
			Unlike the traditional definition of a problem, the input of a big data computing problem must be big data, where big data usually 
			has size greater than or equal to 1 PB.
			The problem defined above is often called as big data computing problem.
			The big data set in the input of a big data computing problem may consists of multiple big data sets.
			The problems discussed in the rest of paper are big data computing problems, and we will simply call them problems in the rest of the paper.
          	
	\section{Pure-tractable Classes}
	
		In this section, we first give the formal definitions of the pure-tractable classes $\mathrm{PL}$ and $\mathrm{ST}$, and then investigate the structure of the pure-tractable classes.
		
		\subsection{Polylogarithmic-tractable Class $\mathrm{PL}$}
	
			As mentioned in \cite{van1990handbook}, the class $\mathrm{DLOGTIME}$ consists of all problems that can be solved by a RATM in $O(\log{n})$ time, where $n$ is the length of the input.
			$\mathrm{DLOGTIME}$ was underestimated before \cite{barrington1990uniformity} \cite{Buss1987The}.
			However, it is very impotent in big data computing \cite{Fan2013Making} and there are indeed many interesting problems in this class \cite{barrington1990uniformity}.
			In this section, we propose the complexity class $\mathrm{PL}$ by extending the $\mathrm{DLOGTIME}$ to characterize problems that are tractable on big data, and inspired by $\mathrm{DLOGTIME}$ and $\mathrm{NC}$ hierarchy, we use $\mathrm{PL}^i$ to reinterpret $\mathrm{PL}$ as a hierarchy.
			
			\begin{definition}
				The class $\mathrm{PL}$ consists of decision problems that can be solved by a \emph{RATM} in polylogarithmic time.
			\end{definition}
			
			\begin{definition}
				For each $i \ge 1$, the class $\mathrm{PL}^i$ consists of decision problems that can be solved by a \emph{RATM} in $O(\log^i{n})$, where $n$ is the length of the input.
			\end{definition}
			
			According to the definition, $\mathrm{PL}^1$ is equivalent to $\mathrm{DLOGTIME}$.
			It is clear that $\mathrm{PL}^1 \subseteq \mathrm{PL}^2 \subseteq \cdots \subseteq \mathrm{PL}^i \subseteq \cdots \subseteq \mathrm{PL}$, which forms the $\mathrm{PL}$ hierarchy. 
			The following Theorem \ref{thrm:pl-pl} shows that $\mathrm{PL}^i \subsetneq \mathrm{PL}^{i+1}$ for $i \in \mathbb{N}$.
			
			\begin{lemma} {\normalfont \cite{dowd1986notes}}\label{lema:dlg-len}
				There is a logarithmic time \emph{RATM}, which takes $x$ as input and generates the output $n$ encoded in binary such that $n = |x|$.
			\end{lemma}
					
			\begin{theorem}\label{thrm:pl-pl}
				For any $i \in \mathbb{N}, \mathrm{PL}^i \subsetneq \mathrm{PL}^{i+1}$.
			\end{theorem}
			
			\begin{proof}
				We prove this theorem by constructing a RATM $M^*$ such that $L(M^*) \in \mathrm{PL}^{i+1} - \mathrm{PL}^i$.
				
				According to Lemma \ref{lema:dlg-len}, there exists a RATM $M_1$, which takes $x$ as input and outputs the binary form of $n = |x|$ in $c\log{n}$ time.
				
				Since $n^{i+1}$ is a polynomial time constructible function for any $i \in \mathbb{N}$, there exists a DTM $M_2$ that takes $x$, whose length is $n$, as input and outputs the binary form of $n^{i+1}$ in time $n^{i+1}$.
				
				By combining $M_1$ and $M_2$, we can construct a RATM $M$ that works as follows.
				Given an input $x$, $M$ first simulates $M_1$ on $x$ and outputs the binary form of $n = |x|$. 
				Then, $M$ simulates $M_2$'s on the binary form of $n = |x|$.
				The total running time of $M$ is $\log^{i+1}{n}$.
				
				Now, we are ready to construct $M^*$.
				On any input $x$, $M^*$ works as follows: 
				
				(1) $M^*$ simulates the computation of $M$ on input $x$ and the computation of $\mathcal{U}$ on input $x$ simultaneously. $\mathcal{U}$ is slightly changed at the process of parsing the input, that is $\mathcal{U}$ copies each symbol of $x$ to its first extra tape as the transition functions until it reads $3$ continuous $1$. Then it acts the same as introduced in Theorem \ref{thrm:uratm}.
				
				(2) Any one of $\mathcal{U}$ and $M$ halts, $M^*$ halts, and the state entered by $M^*$ is determined as follow:
				
				(a) If $\mathcal{U}$ halts first and enters the accept state, then $M^*$ halts and enters the reject state.
				
				(b) If $\mathcal{U}$ halts first and enters the reject state, then $M^*$ halts and enters the accept state.
				
				(c) If $M$ halts first and enters state $q$, then $M^*$ halts and enters $q$.
				
				The running time of $M^*$ is at most $\log^{i+1}{n}$, so $L(M^*) \in \mathrm{PL}^{i+1}$.
				
				Assume that there is a $\log^i{n}$ time RATM $N$ such that $L(N) = L(M^*)$.
				Since $\lim\limits_{n \to \infty} \frac{\log^i{n}\log{\log^i{n}}}{\log^{i+1}{n}} = 0$, there must be a number $n_0$ such that $\log^i{n}\log{\log^i{n}} < \log^{i+1}{n}$ for each $n \ge n_0$.
				Let $x$ be a string representing the machine $N$ whose length is at least $n_0$. Such string exists since a string of a RATM can be added any long string behind the encoding of the RATM.
				We have $Time_\mathcal{U}(x, x) \le cTime_{N}(x)\log{(Time_{N}(x))} \le c\log^i{n}\log{\log^i{n}} \le \log^{i+1}{n}$.
				It means that $\mathcal{U}$ halts before $M$, which is in contradiction with (a) and (b) of $M^*$'s work procedure.  
				\qed	
			\end{proof}
		
		\subsection{Sublinear-tractable Class $\mathrm{ST}$}
	
			To denote all problems can be solve in sublinear time, the complexity class $\mathrm{ST}$ is proposed in this subsection.
			And the relation between $\mathrm{ST}$ and $\mathrm{PL}$ is investigated.
			We first give the formal defintiion of $\mathrm{ST}$.
			
			\begin{definition}
				The class $\mathrm{ST}$ consists of the decision problems that can be solved by a \emph{RATM} in $o(n)$ time, where $n$ is the size of the input.
			\end{definition}
			
			There are indeed many problems that can solved in $o(n)$ time, such as searching in a sorted list, point location in two-dimensional Delaunay triangulations, and checking whether two convex polygons intersect mentioned in \cite{chazelle2005sublinear} \cite{Czumaj2010Sublinear} \cite{Rubinfeld2011Sublinear}.

			To understand the structure of pure-tractable classes, we study the relation between $\mathrm{ST}$ and $\mathrm{PL}$.
			Theorem \ref{thrm:pl-st} shows that $\mathrm{ST}$ contains $\mathrm{PL}$ properly.
			This result indicates that there is a gap between polylogarithmic time class and linear time class.
			
			\begin{theorem}\label{thrm:pl-st}
				$\mathrm{PL} \subsetneq \mathrm{ST}$
			\end{theorem}
		
			\begin{proof}
				First, we define RATIME($t$) to be the class of problems that can be solved by a RATM in time $O(t(n))$.
				It is obviously that $\mathrm{PL}^i \subseteq$ RATIME($\sqrt{n}\log^2{n}$).
				Hence, $\mathrm{PL} = \bigcup_{i \in \mathbb{N}}\mathrm{PL}^i \subseteq \bigcup_{i \in \mathbb{N}}$ RATIME($\sqrt{n}\log^2{n}$) = RATIME($\sqrt{n}\log^2{n}$).
				Similar to the proof of Theorem 2, we show that there is a RATM $M^*$ such that $L(M^*) \in$ RATIME($\sqrt{n}\log^2{n}$) $- \mathrm{PL}$.
				
				We first construct a RATM $M$, which halts in $\Theta(\sqrt{n}\log^2{n})$ steps on input $x$ with size $n$.
				$M$ first computes binary form of $n=|x|$ according to Lemma \ref{lema:dlg-len}.
				Then $M$ enumerates binary number from $1$ to$\sqrt{n}$.
				In the $i$th enumeration step, the binary number $i$ is first enumerated.
				Then, $M$ computes $i \times i$ and makes a comparison with $n$.
				$M$ halts if and only if $i \times i > n$.
				The maximum number enumerated by $M$ is $\sqrt{n}$, the enumeration takes $\log{n}$ time, the multiplication takes $\log^2{n}$ time and the comparison takes $\log{n}$ time.
				So the running time of $M$ is $\Theta(\sqrt{n}\log^2{n})$.
				
				We construct $M^*$ as follows. On an input $x$,
				
				(1) $M^*$ simulates the computation of $M$ on input $x$ and the computation of $\mathcal{U}$ on input $x$ simultaneously. $\mathcal{U}$ is slightly changed at the process of parsing the input, that is $\mathcal{U}$ copies each symbol of $x$ to its first extra tape as the transition functions until it reads $3$ continuous $1$. Then it acts the same as introduced in Theorem \ref{thrm:uratm}.
				
				(2) Any one of $\mathcal{U}$ and $M$ halts, $M^*$ halts, and the halt state of $M^*$ is determined as follow:
				
				(a) If $\mathcal{U}$ halts first and enters the accept state, then $M^*$ halts and enters the reject state.
				
				(b) If $\mathcal{U}$ halts first and enters the reject state, then $M^*$ halts and enters the accept state.
				
				(c) If	$M$ halts first and enters state $q$, then $M^*$ halts and enters state $q$.
				
				The running time of $M^*$ is at most $\sqrt{n}\log^2{n}$, so $L(M^*) \in$ RATIME($\sqrt{n}\log^2{n}$).
				
				Assume that for some $i \in \mathbb{N}$ there is a $\log^i{n}$ time RATM $N$ such that $L(N) = L(M^*)$.
				Since $\lim\limits_{n \to \infty} \frac{\log^i{n}\log{\log^i{n}}}{\sqrt{n}\log^2{n}} = 0$, there must be a number $n_0$ such that $\log^i{n}\log{\log^i{n}} < \sqrt{n}\log^2{n}$ for each $n \ge n_0$.
				Then, let $x$ be a string representing the machine $N$ whose length is at least $n_0$.
				Such string exists since a string of a RATM can be added any long string behind the encoding of the RATM.
				We have $Time_\mathcal{U}(x, x) \le cTime_{N}(x)\log{(Time_{N}(x))} \le c\log^i{n}\log{\log^i{n}} \le \sqrt{n}\log^2{n}$.
				It means that $\mathcal{U}$ halts before $M$, which generates a contradiction with (a) and (b) of $M^*$'s work procedure. \qed
			\end{proof}
			
		\subsection{Reduction and Complete Problems}
			
			In this subsection, we first give the definition of $\mathrm{DLOGTIME}$ reduction \cite{Buss1987The}. Then, it is proved that $\mathrm{PL}$ and $\mathrm{ST}$ is closed under the $\mathrm{DLOGTIME}$ reduction.
			Moreover, the $\mathrm{PL}$-completeness and $\mathrm{ST}$-completeness are defined under $\mathrm{DLOGTIME}$ reduction.
			Finally, a $\mathrm{PL}$-complete problem and a $\mathrm{ST}$-complete problem are given.
			
			\begin{definition}{\normalfont \cite{Buss1987The}}	
				A polynomial reduction $f$ from a problem $A$ to a problem $B$ is a $\mathrm{DLOGTIME}$ reduction if the language $\{(x, i, c) :$ the $i$th bit of $f(x)$ is $c\}$ is in $\mathrm{DLOGTIME}$. 
			\end{definition}
		
			The definition of $\mathrm{DLOGTIME}$ reduction is different from the reductions what we use before.
			It requires that the checking time of a specific location is logarithmic, but the total time of reduction is bounded by polynomial time.
			The following two theorems show that $\mathrm{PL}$ and $\mathrm{ST}$ are closed under $\mathrm{DLOGTIME}$ reduction.
		
			\begin{theorem}\label{thrm:dlg-pl}
				If $B \in \mathrm{PL}^i$ and there is a $\mathrm{DLOGTIME}$ reduction from $A$ to $B$ then $A \in \mathrm{PL}^{i+1}$.
			\end{theorem}
		
			\begin{proof}
				Let $M_B$ be the RATM which solves $B$ in time $O(\log^i{n})$.
				Let $f$ be the $\mathrm{DLOGTIME}$ reduction from $A$ to $B$.
				We construct a RATM $M_A$ which solves $A$ in $O(\log^{i+1}{n})$ time.
				On an input $x \in A$, $M_A$ simulates $M_B$ as follows.
				
				(1) For moves of $M_B$, which do not read input tape, $M_A$ directly simulates $M$.
				
				(2) For moves of $M_B$, which read input tape, assuming the $i$th symbol of input is reading, $M_A$ checks whether the $i$th symbol of $f(x)$ is $a$ for each $a$ in input tape symbol set of $M_B$. 
				When it finds the correct symbol $a$, it continues the simulation of $M_B$.
				
				Assume that $M_B$ stops on input $x$ in $\log^i{n} = k + l$ steps, in which $k$ steps do not read input tape, and $l$ steps need to read input tape.  
				
				To simulate the $k$ steps, which do not read input tape, $M_A$ needs $k \leq \log^i{n}$ steps.
				
				Assume that the $i$th step of $M_B$ read input tape.
				To simulate the $i$th step, $M_A$ needs to get the input symbol $a$ in $f(x)$ by scanning the input tape symbol set of $M_B$, which needs $c\log{n}$ steps by the definition of the $\mathrm{DLOGTIME}$ reduction, where $c$ is the size of input tape symbol set of $M_B$. 
				Thus, to simulate the $l$ steps of $M_B$, which read input tape, $M_A$ needs $lc\log{n} \leq c\log^i{n}\log{n}$ since $l \leq \log^i{n}$.
				
				In summary, $M_A$ needs at most $\log^i{n} + c\log^i{n}\log{n}=O(\log^i{n}\log{n})$, that is $O(\log^{i+1}{n})$ steps.
				\qed
			\end{proof}

            From the theorem \ref{thrm:dlg-pl}, we can directly derive the following corollary \ref{coro:pl-close}.		
			
			\begin{corollary}\label{coro:pl-close}
				 $\mathrm{PL}$ is closed under $\mathrm{DLOGTIME}$ reduction.
			\end{corollary}
		
			\begin{theorem}\label{thrm:st-close}
				$\mathrm{ST}$ is closed under $\mathrm{DLOGTIME}$ reduction.
			\end{theorem}
		
			\begin{proof}
				Let $M_B$ be the RATM which solves $B$ in time $T(n) = o(n)$.
				Let $f$ be the $\mathrm{DLOGTIME}$ reduction from $A$ to $B$.
				We construct a RATM $M_A$ which solves $A$ in $o(n)$ time.
				On an input $x \in A$, $M_A$ simulates $M_B$ as follows.
				
				(1) For moves of $M_B$, which do not read input tape, $M_A$ directly simulates $M$.		
				
				(2) For moves of $M_B$, which read input tape, assuming $i$th symbol of input is reading, $M_A$ checks whether the $i$th symbol of $f(x)$ is $a$ for each $a$ in input tape symbol of $M_B$.
				When it finds the correct symbol $a$, it continues the simulation of $M_B$.
				
				Assume that $M_B$ stops on input $x$ in $T(n) = k + l$ steps, in which $k$ steps do not read input tape, and $l$ steps need to read input tape.
				
				To simulate the $k$ steps, which do not read input tape, $M_A$ needs $k \le T(n)$ steps.
				
				Assume that the $i$th step of $M_B$ read input tape.
				To simulate the $i$th step, $M_A$ needs tp get the input symbol $a$ in $f(x)$ by scanning the input tape symbol set of $M_B$, which needs $c\log{n}$ steps by the definition of the $\mathrm{DLOGTIME}$ reduction,
				where $c$ is the size of input tape symbol set of $M_B$.
				Thus, to simulate the $l$ steps of $M_B$, which read input tape, $M_A$ needs $lc\log{n} \le c\log{n}T(n)$ since $l \le T(n)$.
				
				In summary, $M_A$ needs at most $T(n) + cT(n)\log{n} = O(T(n)\log{n})$, that is $o(n)$ steps.			\qed
			\end{proof}
			
			The definitions of $\mathrm{PL}$-completeness and $\mathrm{ST}$-completeness under $\mathrm{DLOGTIME}$ reduction are given in the following.
			
			\begin{definition}
				A problem $L$ is $\mathrm{PL}$-hard under $\mathrm{DLOGTIME}$ reduction if there is a $\mathrm{DLOGTIME}$ reduction from $L'$ to $L$ for all $L'$ in $\mathrm{PL}$.
				A problem $L$ is $\mathrm{PL}$-complete under $\mathrm{DLOGTIME}$ reduction if $L \in \mathrm{PL}$ and $L$ is $\mathrm{PL}$-hard.
			\end{definition}
			
			\begin{definition}
				A problem $L$ is $\mathrm{ST}$-hard under $\mathrm{DLOGTIME}$ reduction if there is a $\mathrm{DLOGTIME}$ reduction from $L'$ to $L$ for all $L'$ in $\mathrm{ST}$.
				A problem $L$ is $\mathrm{ST}$-complete under $\mathrm{DLOGTIME}$ reduction if $L \in \mathrm{ST}$ and $L$ is $\mathrm{ST}$-hard.
			\end{definition}

			Bounded Halting Problem (BHP) is $\mathrm{NP}$-complete.
			We show a sublinear version of BHP, and prove that it is $\mathrm{PL}$-complete and $\mathrm{ST}$-complete.
			
			Sublinaer Bounded Halting problem (SBHP) : 
			
			INPUT: the code $c(M)$ of a RATM $M$, $M$'s input $y$ and a string $0^t$,  where $t \leq \log^{i+1}{|y|}$.
			
			OUTPUT: Does machine $M$ accepts $y$ within $t$ moves?
			
			\begin{theorem}
				\emph{SBHP} is $\mathrm{PL}$-complete.
			\end{theorem}
		
			\begin{proof}
				First, we show SBHP is in $\mathrm{PL}$.
				The URATM $\mathcal{U}$ can be used to simulate $t$ steps of $M$ on input $y$, if $M$ enters accept state, $\mathcal{U}$ accepts $y$, else $\mathcal{U}$ rejects $y$.
				If the running time of $M$ on input $y$ is $\log^i{|y|}$ for some $i \in \mathbb{N}$, the running time of $\mathcal{U}$ is $\min\{t, \log^i{|y|}\log{\log^i{|y|}}\}$ according to Theorem \ref{thrm:uratm}. Thus SBHP is in $\mathrm{PL}$ since $\min\{t,\log^i{|y|}\log{\log^i{|y|}}\} \leq \log^{i+1}{|y|}$
				
				Next, we prove that there is a $\mathrm{DLOGTIME}$ reduction from $L$ to SBHP for all $L \in \mathrm{PL}$.
				Since $L \in \mathrm{PL}$, there exist a RATM $M$ and an integer $i \in \mathbb{N}$ such that $M$ accepts $y$ in time $\log^i{|y|}$ if and only if $y \in L$.
				It is simple to transform any $y \in L$ to an instance in SBHP by letting $f(y) = \langle c(M), y, 0^{\log^{i+1}{|y|}} \rangle$, where $c(M)$ is the code of $M$.
				
				It remains to prove that there is RATM $N$ such that  
				$N$ can decide whether the $i$th bit of $f(y)$ is a symbol $c$ 
				in $\log{n}$ time. $N$ works as follows.
				
				(1) $N$ computes the length of $c(M)$ and $y$ and stores them in binary.
				
				(2) $N$ calculates $|c(M)| + |y|$.
				
				(3) $N$ determines the $i$th symbol of $f(y)$ as follows.
				
				$~~~~~$ If $i \le |c(M)|$, $N$ outputs the $i$th symbol of $c(M)$.
				
				$~~~~~$ If $|c(M)| < i \le |c(M)| + |y|$, $N$ outputs the $(i - |c(M)|)$th symbol of $y$.
				
				$~~~~~$ If $i > |c(M)| + |y|$, $N$ outputs $0$.
				
				Scince all numbers are encoded in binary, the $N$ needs at most $\max\{\log{|c(M)|},\\ \log{|y|}\}$.
				
				It is clear that $f$ is computable in polynomial time.
				\qed
			\end{proof}
			
			By changing $t \le \log^{i+1}{n}$ to $t \le T(|y|)\log{T(|y|)}$ of SBHP, where $T$ is the running time of $M$ on input $y$, we derive the following theorem.
			
			\begin{theorem}
				\emph{SBHP} is $\mathrm{ST}$-complete.
			\end{theorem}
		
			\begin{proof}
				First, we prove that SBHP is in $\mathrm{ST}$.
				The URATM $\mathcal{U}$ can be used to simulate $t$ steps of $M$ on input $y$, if $M$ enters accept state, $\mathcal{U}$ accepts $y$, else $\mathcal{U}$ rejects $y$.
				If the running time of $M$ on input $y$ is $T(|y|)$ for some $T$ in $o(n)$, the running time of $\mathcal{U}$ is $\min\{t, T(|y|)\log{T(|y|)}\}$ according to Theorem \ref{thrm:uratm}.
				Thus, SBHP is in $\mathrm{ST}$ since $\min\{t, T(|y|)\log{T(|y|)}\} \in o(|y|)$.
				
				Next, we prove that there is a $\mathrm{DLOGTIME}$ reduction from $L$ to SBHP for all $L \in \mathrm{ST}$.
				Since $L \in \mathrm{ST}$, there exist a RATM $M$ and a function $T \in o(n)$ such that $M$ accepts $y$ in time $T(|y|)$ if and only if $y \in L$.
				It is simple to transform any $y \in L$ to an instance in SBHP by $f(y) = \langle c(M), y, 0^{T(|y|)\log{T(n)}} \rangle$, where $c(M)$ in the code of $M$.
				
				It remains to prove that there is a RATM $N$ such that $N$ can decide whether the $i$th bit of $f(y)$ is a symbol $c$ in $\log{n}$ time.
				$N$ works as follows.
				
				(1) $N$ computes the length of $c(M)$ and $y$ and stores them in binary.
				
				(2) $N$ calculates $|c(M)| + |y|$.
				
				(3) $N$ determines the $i$th symbol of $f(y)$ as follows.
				
				$~~~~~$ If $i \le |c(M)|$, $N$ outputs the $i$th symbol of $c(M)$.
				
				$~~~~~$ If $|c(M)| < i \le |c(M)| + |y|$, $N$ outputs the $(i - |c(M)|)$th symbol of $y$.
				
				$~~~~~$ If $i > |c(M)| + |y|$, $N$ outputs $0$.
				
				Since all numbers are encoded in binary, $N$ needs at most $\max\{\log{|c(M)}|,\\ \log{|y|}\}$ steps.
				
				It is clear that $f$ is computable in polynomial time.
				\qed
			\end{proof}
	
	\section{Pseudo-tractable Classes}
		
		In this section, we study the big data computing problems that can be solved in sublinear time after a PTIME preprocessing. 	
		We propose two complexity classes, $\mathrm{PTR}$ and $\mathrm{PTE}$, and investigate the relations among $\mathrm{PTR}$, $\mathrm{PTE}$ and other complexity classes.
	    For easy to understand, we first review the definition of a problem in big data computing, that is,
	    
	    INPUT: big data $D$, and a function $\mathcal{F}$ on $D$.
	    
	    OUTPUT: $\mathcal{F}(D)$.
		
		\subsection{Pseudo-tractable Class by Reducing $|D|$}
		
			We will use $\mathrm{PTR}$ to express the pseudo-tractable class by reducing $|D|$, which is defined as follows.
			
			\begin{definition}
				A problem $\mathcal{F}$ is in the complexity class $\mathrm{PTR}$ if there exists a \emph{PTIME} preprocessing function $\Pi$ such that for big data $D$,
				
				(1) $|\Pi(D)| < |D|$ and $\mathcal{F}(\Pi(D)) = \mathcal{F}(D)$.
				
				(2) $\mathcal{F}(\Pi(D))$ can be solved by a \emph{RATM} in $o(|D|)$ time.	
			\end{definition}
		
			Data $D$ is preprocessed by a preprocessing function $\Pi$ in polynomial time to generate another structure $\Pi(D)$. 
			Besides PTIME restriction on $\Pi$, it is required that the size of $\Pi(D)$ is smaller than $D$.
			This allows many of previous polynomial time algorithms to be used.
			For example, $\mathcal{F}(\Pi(D))$ can be solved by a quadratic polynomial time algorithm if $|D| = n$ and $|\Pi(D)| = n^{1/3}$, and the time needed for solving $\mathcal{F}(\Pi(D))$ is $O(n^{2/3}) \in o(n)$.
			To make $\Pi(D)$ less than $D$, $\Pi$ can be data compression, sampling, etc.  
						
			The following simple propositions show the time complexity of the problem after preprocessing.
			
			\begin{proposition}
				If there is a preprocessing function $\Pi$ and a constant $c > 1$ such that $|\Pi(D)| = |D|^{1/c}$ for any $D$, and there is a algorithm to solve $\mathcal{F}(\Pi(D))$ in time of polynomials of degree $d$, where $d < c$, then $F(\Pi(D))$ can be solved in $o(|D|)$ time, and thus $\mathcal{F}$ is in $\mathrm{PTR}$.
			\end{proposition}
		
			\begin{proposition}
				If there is a preprocessing function $\Pi$ and a constant $c$ such that $|\Pi(D)| = \log^c{|D|}$ for any $D$, and there is a \emph{PTIME} algorithm to solve $\mathcal{F}(\Pi(D))$, then $\mathcal{F}(\Pi(D))$ can be solved in $o(|D|)$ time, and thus $\mathcal{F}$ is in $\mathrm{PTR}$. 
			\end{proposition}
		
			\begin{proposition}
				If there is a preprocessing function $\Pi$ and a constant $c \in (0, 1)$ such that  $|\Pi(D)| = c\log{|D|}$ for any $D$, and there is a $O(2^{|\Pi(D)|})$ time algorithm to solve $\mathcal{F}(\Pi(D))$, then $\mathcal{F}(\Pi(D))$ can be solved in $o(|D|)$ time, and thus $\mathcal{F}$ is in $\mathrm{PTR}$.
			\end{proposition}
			
		\subsection{Pseudo-tractable Class by Extending $|D|$ }
			
			Obviously, $\mathrm{PTR}$ does not characterize all problems that can be solved in sublinear time after preprocessing.
			We define another complexity class $\mathrm{PTE}$ to denote remaining problems that can be solved in sublinear time after preprocessing.
			
			\begin{definition}
				A problem $\mathcal{F}$ is in the complexity class $\mathrm{PTE}$ if there exists a \emph{PTIME} preprocessing $\Pi$ such that for big data $D$,
				
				(1) $|\Pi(D)| \ge |D|$ and $\mathcal{F}(\Pi(D)) = \mathcal{F}(D)$.
				
				(2) $\mathcal{F}(\Pi(D))$ can be solved by a \emph{RATM} in $o(|D|)$ time.	
			\end{definition}
		
			The only difference between $\mathrm{PTE}$ and $\mathrm{PTR}$ is that the former requires that preprocessing results in a larger dataset compared to original dataset.
			Intuitively, if a problem $\mathcal{F}$ is in $\mathrm{PTE}$, then $\mathcal{F}$ can be solved in sublinear time by sacrificing space. 
			For example, many queries are solvable in sublinear time by building index before query processing. 
			Many non-trivial data structures are also designed for some problems to accelerate the computation.
			Extremely, all the problems in $\mathrm{P}$ can be solved in polynomial time and stored previously, and can be solved in $O(1)$ time later, that is, 
			$\mathrm{P} \subseteq \mathrm{PTE}$.  
			
			There are also some propositions show the complexity class of the problem after preprocessing.

			\begin{proposition}
				If there is a preprocessing function $\Pi$ and a constant $c \ge 1$ such that $|\Pi(D)| = c|D|$ for any $D$, and there is a $o(\Pi(D))$ time algorithm to solve $\mathcal{F}(\Pi(D))$, then $\mathcal{F}(\Pi(D))$ can be solved in $o(|D|)$ time, and thus $\mathcal{F}$ is in $\mathrm{PTE}$.
			\end{proposition}
		
			\begin{proposition}
				If there is a preprocessing function $\Pi$ and a constant $c \ge 1$ such that $|\Pi(D)| = |D|^c$ for any $D$, and there is a $O(\log^i(\Pi(d)))$ time algorithm to solve $\mathcal{F}(\Pi(D))$, then $\mathcal{F}(\Pi(D))$ can be solved in $o(|D|)$ time, and thus $\mathcal{F}$ is in $\mathrm{PTE}$.
			\end{proposition}
		
		\subsection{Relations between Pseudo-tractable Classes and Other Classes}
			
			In the rest of the paper, we use $\mathrm{PT}$ to express $\mathrm{PTR} \cup \mathrm{PTE}$
			 
			\begin{theorem}\label{thrm:pt-p}
				$\mathrm{PT} \subseteq \mathrm{P}$
			\end{theorem}
		
			\begin{proof}
				For every problem $\mathcal{F} \in \mathrm{PT}$, there exists a polynomial time DTM $M_1$ to generate $\Pi(D)$, and there exists a $o(|D|)$ time RATM $M_2$ to compute $F(\Pi(D)) = \mathcal{F}(D)$.
				Then, we can get a polynomial time RATM $M$ by combining $M_1$ and $M_2$, which take $D$ as input and outputs $\mathcal{F}(D)$.
				According to Lemma \ref{lema:ratm-tm}, we can derive a polynomial time DTM that takes $D$ as input and outputs $\mathcal{F}(D)$.
				Thus, $\mathcal{F} \in \mathrm{P}$.
				\qed
			\end{proof}
			
			\begin{corollary}
				No $\mathrm{NP}$-Complete problem is in $\mathrm{PT}$ if $\mathrm{P} \ne \mathrm{NP}$.
			\end{corollary}
		
			\begin{proof}
				If a $\mathrm{NP}$-complete problem is in $\mathrm{PT}$, then the problem is in $\mathrm{P}$ according to theorem \ref{thrm:pt-p}.
				Thus, $\mathrm{P} = \mathrm{NP}$, which is contradict to $\mathrm{P} \ne \mathrm{NP}$.
				\qed
			\end{proof}
		
			Recalling the definitions of $\mathrm{PTR}$ and $\mathrm{PTE}$, $\mathrm{PTR}$ and $\mathrm{PTE}$ concern all big data computing problems rather than only queries.
			To investigate the relations among $\mathrm{PTR}$, $\mathrm{PTE}$, and $\sqcap'\mathrm{T^0_Q}$ \cite{Yang2017Tractable}, we restrict $\mathcal{F}$ in the definitions of $\mathrm{PTR}$ and $\mathrm{PTE}$ to boolean query class $\mathcal{Q}$ to define two pseudo-tractable query complexity classes in the following.
			
			Following the conventions in \cite{Fan2013Making}, a boolean query class $\mathcal{Q}$ can be encoded as $S=\{\langle D, Q \rangle\} \subseteq \Sigma^* \times \Sigma^*$ such that $\langle D, Q \rangle \in S$ if and only if $Q(D)$ is true, where $Q \in \mathcal{Q}$ is defined on $D$.
			
			Now we define $\mathrm{PTR^0_Q}$ and $\mathrm{PTE^0_Q}$ as follows.
			
			\begin{definition}
				A language $S = \{\langle D, Q\rangle \}$ is \emph{pseudo-tractable by reducing datasets} if there exist a \emph{PTIME} preprocessing function $\Pi$ and a language $S'$ such that for all $D, Q$,
				
				(1) $\langle D, Q \rangle \in S$ iff $\langle \Pi(D), Q \rangle \in S'$, $|\Pi(D)| < |D|$ and
				
				(2) $S'$ can be solved by a RATM in $o(|D|)$ time.
			\end{definition}
			
			A class of queries, $\mathcal{Q}$, is \emph{pseudo-tractable by reducing datasets} if $S$ is pseudo-tractable by reducing datasets, where $S$ is the language for $\mathcal{Q}$.
			
            \begin{definition}
				The query complexity class $\mathrm{PTR^0_Q}$ is the set of $\mathcal{Q}$ that is pseudo-tractable by reducing datasets.
		    \end{definition}
			
			\begin{definition}
				A language $S = \{\langle D, Q\rangle \}$ is \emph{pseudo-tractable by extending datasets} if there exist a \emph{PTIME} preprocessing function $\Pi$ and a language $S'$ of pairs such that for all $D, Q$,
				
				(1) $\langle D, Q \rangle \in S$ iff $\langle \Pi(D), Q \rangle \in S'$ , $|\Pi(D)| \ge |D|$, and,
				
				(2) $S'$ can be solved by a RATM in $o(|D|)$ time.
			\end{definition}
		
			A class of queries $\mathcal{Q}$ is \emph{pseudo-tractable by extending datasets} if $S$ is extended pseudo-tractable, where $S$ is the language for $\mathcal{Q}$.
			
			\begin{definition}
				The query complexity class $\mathrm{PTE^0_Q}$ is the set of $\mathcal{Q}$ that is pseudo-tractable by extending datasets.
			\end{definition}
		
		    Obviously, $\mathrm{PTR^0_Q} \subseteq \mathrm{PTR}$ and 
		    $\mathrm{PTE^0_Q} \subseteq \mathrm{PTE}$. 
			
			The following theorem shows that $\sqcap'\mathrm{T^0_Q}$ is strictly smaller than $\mathrm{PTR^0_Q}$ and also indicates that the logarithmic restriction on output size of preprocessing \cite{Yang2017Tractable} is too strict.
			
			\begin{theorem}
				$\sqcap'\mathrm{T^0_Q} \subsetneq \mathrm{PTR^0_Q}$.
			\end{theorem}
		
			\begin{proof}
				Let $\mathcal{Q} \in \sqcap'\mathrm{T^0_Q}$.
				For any $\langle D, Q \rangle \in S_{\mathcal{Q}}$, $Q$ satisfies short-query property, $i.e.$, $|Q| \in O(\log{|D|})$,   
				there is a PTIME preprocessing function $\Pi$ such that $|\Pi(D)| \in O(\log{|D|})$ and a language $S'$ such that $\langle D, Q\rangle \in S$ if and only if $\langle \Pi(D), Q\rangle \in S'$, and 
				$S'$ is in $\mathrm{P}$, that is, the time of computing $Q(D')$ is $O((|D'|)^c) = O(\log^c{|D|})$, according to the definition of $\sqcap'\mathrm{T^0_Q}$. 
				Thus, $\sqcap'\mathrm{T^0_Q} \subseteq \mathrm{PTR^0_Q}$ since $O(\log^c{|D|}) \subseteq o(|D|)$.
				
				It has been proved that the boolean queries for breadth–depth search, denoted by $\mathcal{Q}_{\rm{BDS}}$, is not in $\sqcap'\mathrm{T^0_Q}$  \cite{Yang2017Tractable}. We need to prove that $\mathcal{Q}_{\rm{BDS}}$ is in $\mathrm{PTR^0_Q}$.

				Given an undirected graph $G = (V, E)$ with $n$ nodes with numbers and $m$ edges and a pair of nodes $u$ and $v$ in $V$, the query is whether $u$ visited before $v$ in the breadth-depth search of $G$. The query is processed as follows. 
				
				We first define preprocessing funciton $\Pi$ that performs breadth-depth search on $G$ \cite{Greenlaw1993BREADTH}, and return a array $M$, where the index of $M$ is the node number 
				and if $M[i]<M[j]$ then node $i$ is visited before node $j$.  
				
				After the preprocessing, the query can be answered in two following steps:
				(1) Access $M[u]$ and $M[v]$.
				(2) If $M[u]<M[v]$ then $u$ is visited before $v$ else $v$ is visited before $u$.
				
				It is obvious that the query can be answered by a RATM in $O(\log n)$ time. 
				
				Moreover, the size of $G$ encoded by adjacency list is $O(n^2)$ and the size of $M$ is $O(n\log{n})$, that is, $|M| < |G|$.
				
				Therefore, $\mathcal{Q}_{\rm{BDS}} \in \mathrm{PTR^0_Q}$.
				Consequently, $\sqcap'\mathrm{T^0_Q} \subsetneq \mathrm{PTR^0_Q}$.
				\qed
			\end{proof}
		
			Let $\mathrm{PT^0_Q}=\mathrm{PTR^0_Q} \cup \mathrm{PTE^0_Q}$.
			We extend the definition of $\mathrm{PT^0_Q}$ to the decision problems by 
			using the definition of factorization in \cite{Fan2013Making}.
			A factorization of a language $L$ is $\Upsilon = (\pi_1, \pi_2, \rho)$, where $\pi_1$, $\pi_2$, and $\rho$ are in $\mathrm{NC}$ and satisfy that  $\langle \pi_1(x), \pi_2(x) \rangle \in S_{(L, \Upsilon)}$ and $\rho(\pi_1(x), \pi_2(x)) = x$ for all $x \in L$.
			
			\begin{definition}\label{def:made-pt}
				A decision problem $L$ \emph{can be made pseudo-tractable} if there exists a factorization $\Upsilon$ of $L$ such that the language $S_{(L, \Upsilon)}$ of pairs for $(L, \Upsilon)$ is pseudo-tractable, $i.e$, $S_{(L, \Upsilon)} \in \mathrm{PT^0_Q}$. 
			\end{definition}
		
		    In the rest of the paper, we use $\mathrm{PT_P}$ to denote the set of all decision problems that can be made pseudo-tractable.
			
			The following theorem tells us that all problem in $\mathrm{P}$ can be made pseudo-tractable.
			
			\begin{theorem}
				$\mathrm{PT_P} = \mathrm{P}$.
			\end{theorem}
		
				\begin{proof}
				First, we prove that $\mathrm{PT_P} \subseteq \mathrm{P}$.
				A factorization $\Upsilon$ of $L$ are three $\mathrm{NC}$ computable functions, and it is known that $\mathrm{NC} \subseteq \mathrm{P}$.
				The preprocessig function and query processing can be computed in polynomial time.
				Thus, $\mathrm{PT_P} \subseteq \mathrm{P}$.
				
				Then, we prove that $\mathrm{P} \subseteq \mathrm{PT_P}$.
				It is proved that \rm{BDS} is $\mathrm{P}$-complete under $\mathrm{NC}$-reduction \cite{Greenlaw1993BREADTH}, $i.e.$, 
				there is a $f(\cdot) \in \mathrm{NC}$ such that $x \in L$ iff $f(x) \in BDS$ for every $L \in \mathrm{P}$.
				Let $\Upsilon_{\rm{BDS}} = (\pi_1, \pi_2, \rho)$ be a factorization of \rm{BDS}, where $\pi_1(x) = G, \pi_2(x) = (u, v)$, and $\rho$ maps $\pi_1(x)$ and $\pi_2(x)$ back to $x$.
				As mentioned in the proof of Theorem 8, $S_{(L_\mathrm{BDS}, \Upsilon_\mathrm{BDS})} \in \mathrm{PTR^0_Q} \in \mathrm{PT^0_Q}$.
				
				For every $L \in P$, we can first use $f(\cdot)$ to transform $x \in L$ into a instance of $\mathrm{BDS}$, and then use factorization  $\Upsilon_\mathrm{BDS}$ to construct
				$\Upsilon = (\pi'_1, \pi'_2, \rho')$, where $\pi'_1(x) = \pi_1(f(x))=G$ and $\pi'_2(x) = \pi_2(f(x))=(u, v)$.
				Hence, we get a $\Upsilon$ for $L$ such that $S_{(L, \Upsilon)} = S_{(L_\mathrm{BDS}, \Upsilon_\mathrm{BDS})} \in \mathrm{PT^0_Q}$
				Thus $L \in \mathrm{PT_P}$ by definition \ref{def:made-pt}.
				Therefore, $\mathrm{P} \subseteq \mathrm{PT_P}$.
				\qed
			\end{proof}

	\section{Conclusion}
	
		This work aims to recognize the tractability in big data computing.
		The RATM is used as the computational model in our work, and an efficient URATM is devised.
		
		Two pure-tractable classes are proposed.
		One is the class $\mathrm{PL}$ consisting of the problems that can be solved in polylogarithmic time by a RATM.
		The another one is the class $\mathrm{ST}$ including all the problems that can be solved in sublinear time by a RATM.
		The structure of pure-tractable classes is deeply investigated.
		The polylogarithmic time hierarchy, $\mathrm{PL}^i \subsetneq \mathrm{PL}^{i+1}$, is first proved. Then, is proved that $\mathrm{PL} \subsetneq \mathrm{ST}$.
		Finally, $\mathrm{DLOGTIME}$ reduction is proved to be closed for $\mathrm{PL}$ and $\mathrm{ST}$, and the first $\mathrm{PL}$-complete problem and the first $\mathrm{ST}$-complete problem are given.
		
		Two pseudo-tractable classes, $\mathrm{PTR}$ and $\mathrm{PTE}$, are also proposed.
		$\mathrm{PTR}$ consisting of all the problems that can solved by a RATM in sublinear time after a PTIME preprocessing by reducing the size of input dataset.
		$\mathrm{PTE}$ consisting of all the problems that can solved by a RATM in sublinear time after a PTIME preprocessing by extending the size of input dataset. 
		The relations among the pseudo-tractable classes and other complexity classes are investigated.
		They are proved that $\mathrm{PT} \subseteq \mathrm{P}$ and $\sqcap'\mathrm{T^0_Q} \subsetneq \mathrm{PTR^0_Q}$.
		And for the decision problems, we proved that all decision problems in $\mathrm{P}$ can be made pseudo-tractable.
	
	\section*{Acknowledgment}
		This work was supported by the National Natural Science Foundation of China under grants 61732003, 61832003, 61972110 and U1811461.
		
	\bibliographystyle{plain}
	\bibliography{ref}

\begin{thebibliography}{10}

\bibitem{PCIewiki}
\url{https://www.wikiwand.com/en/PCI_Express#}.

\bibitem{arora2009computational}
Sanjeev Arora and Boaz Barak.
\newblock {\em Computational complexity: a modern approach}.
\newblock Cambridge University Press, 2009.

\bibitem{barrington1990uniformity}
David A~Mix Barrington, Neil Immerman, and Howard Straubing.
\newblock On uniformity within nc1.
\newblock {\em Journal of Computer and System Sciences}, 41(3):274--306, 1990.

\bibitem{Buss1987The}
Samuel~R. Buss.
\newblock The boolean formula value problem is in alogtime.
\newblock In {\em Nineteenth Acm Symposium on Theory of Computing}, 1987.

\bibitem{chazelle2005sublinear}
Bernard Chazelle, Ding Liu, and Avner Magen.
\newblock Sublinear geometric algorithms.
\newblock {\em SIAM Journal on Computing}, 35(3):627--646, 2005.

\bibitem{Czumaj2010Sublinear}
Artur Czumaj and Christian Sohler.
\newblock Sublinear-time algorithms.
\newblock In {\em Property Testing-current Research $\&$ Surveys}, 2010.

\bibitem{dowd1986notes}
M.~DOWD.
\newblock Notes on log space representation.
\newblock 1986.

\bibitem{du2011theory}
Ding-Zhu Du and Ker-I Ko.
\newblock {\em Theory of computational complexity}, volume~58.
\newblock John Wiley \& Sons, 2011.

\bibitem{Fan2013Making}
W.~Fan, F.~Geerts, and F.~Neven.
\newblock Making queries tractable on big data with preprocessing: (through the
  eyes of complexity theory).
\newblock {\em Proceedings of the Vldb Endowment}, 6(9):685--696, 2013.

\bibitem{Greenlaw1993BREADTH}
Raymond Greenlaw.
\newblock Breadth-depth search is $\mathcal{P}$-complete.
\newblock {\em Parallel Processing Letters}, 3(03):209--222, 1993.

\bibitem{Rubinfeld2011Sublinear}
Ronitt Rubinfeld.
\newblock Sublinear time algorithms.
\newblock {\em Marta Sanz Solé}, 34(2):págs. 1095--1110, 2011.

\bibitem{van1990handbook}
Jan Van~Leeuwen and Jan Leeuwen.
\newblock {\em Handbook of theoretical computer science}, volume~1.
\newblock Elsevier, 1990.

\bibitem{Yang2017Tractable}
Jiannan Yang, Hanpin Wang, and Yongzhi Cao.
\newblock Tractable queries on big data via preprocessing with logarithmic-size
  output.
\newblock {\em Knowledge $\&$ Information Systems}, 56(12):1--23, 2017.

\end{thebibliography}
		
\end{document}